\documentclass[letterpaper,USenglish,cleveref, autoref, thm-restate,authorcolumns]{llncs}
\pdfoutput=1

\usepackage{stmaryrd}
\usepackage[linesnumbered,ruled]{algorithm2e} 
\usepackage{amsmath}
\usepackage{amssymb}
\usepackage{amsfonts}
\usepackage{graphicx}
\usepackage{inputenc}

\DeclareMathOperator{\tlU}{\ensuremath\mathbf{U}}
\DeclareMathOperator{\tlX}{\ensuremath\mathbf{X}}
\DeclareMathOperator{\tlR}{\ensuremath\mathbf{R}}
\DeclareMathOperator{\tlF}{\ensuremath\mathbf{F}}
\DeclareMathOperator{\tlG}{\ensuremath\mathbf{G}}

\DeclareMathOperator{\tlWeakX}{\ensuremath\mathbf{\overline{X}}}

\newcommand{\AP}{\mathcal{AP}} 
\newcommand{\PiAP}{\Pi^{\AP}}	
\newcommand{\sat}{\models}
\newcommand{\den}[1]{\llbracket #1 \rrbracket}
\newcommand{\dia}[1] {\left< #1 \right>}
\newcommand{\var}{{\texttt{var}}}
\newcommand{\qc}{{\operatorname{QC}}}

\newcommand{\phivar}{\phi[\var]}

\newcommand{\logictrue}{\mathit{true}}
\newcommand{\logicfalse}{\mathit{false}}

\newcommand{\dquarter}{D_{\operatorname{quarter}}}
\newcommand{\dmonth}{D_{\operatorname{month}}}

\newcommand{\figref}[1]{Figure~\ref{#1}}
\newcommand{\secref}[1]{Section~\ref{#1}}
\newcommand{\defref}[1]{Definition~\ref{#1}}

\begin{document}

\title{Temporal-Logic Query Checking over Finite Data Streams}
\titlerunning{Query Checking over Finite Streams}
\author{Samuel Huang \and Rance Cleaveland}


\authorrunning{S. Huang \and R. Cleaveland} 

\institute{Department of Computer Science, University of Maryland, College Park, MD, USA
  \email{\{srhuang,rance\}@cs.umd.edu}
  }











\maketitle

\begin{abstract}
  This paper describes a technique for inferring temporal-logic properties for
  sets of finite data streams.  Such data streams arise in many domains,
  including server logs, program testing, and financial and marketing data; temporal-logic formulas that are satisfied by all data streams in
  a set can provide insight into the underlying dynamics of the system generating these
  streams.  Our approach makes use of so-called Linear Temporal Logic (LTL) queries,
  which are LTL formulas containing a missing subformula and interpreted over finite data streams.  Solving such a query
  involves computing a subformula that can be inserted into the query so that
  the resulting grounded formula is satisfied by all data streams in the set.
  We describe an automaton-driven approach to solving this query-checking
  problem and demonstrate a working implementation via a pilot study.

\end{abstract}

\newpage
\section{Introduction}

A central problem in system analysis may be phrased as the \emph{behavioral understanding problem}:  given concrete observations of a system's behavior, infer high-level properties characterizing this behavior.  Such properties can be used for a variety of purposes, including system specification, software understanding (when the system in question is software), and root-cause failure analysis.  Several researchers have studied variants of this problem in a several contexts, from software engineering~\cite{ackermann:rv2010} to data mining~\cite{agrawal:icde95} and artificial intelligence~\cite{fradkin:kis2015,georgala:ecai2016}.

This paper considers the following variant of the behavioral understanding problem:  given (1) a finite set of finite-duration observations of system behavior encoded as finite-length \emph{data streams}, and (2) a temporal-logic \emph{query} (``formula with a hole'') interpreted over such streams, infer formulas that, when plugged into the ``hole,'' yield a temporal-logic formula satisfied by all data streams in the given set.  For example, if the query in question has form $\tlG \texttt{var}$, where $\texttt{var}$ is the ``hole'' and $\tlG$ is the  ``always operator,'' then a solution $\phi$ would be a formula that holds at each time point across all data streams.  Other queries can be used to characterize when error conditions will be tripped, or when temporal correspondences hold between different basic properties captured in the data streams.  Using our query-solving technology, an engineer can collect different system executions, which might be in the form of system logs or experimentally observed data, and then pose queries to develop insights into the mechanisms underpinning system behavior.

The rest of this paper is structured as follows.  The next section discusses related work in temporal-logic query checking and defines finite data streams formally.  The section following then gives the variant of temporal logic, Finite Linear Temporal Logic (Finite LTL), that is the basis for our work and formalizes the query-checking problem for this logic over finite sets of data streams.   We then present a construction for computing finite-automaton-like structures from Finite LTL queries, as well as a method for solving these queries with respect to finite sets of data streams.  The paper concludes with a set of preliminary experimental results using a prototype implementation and a discussion of future research directions.


\section{Background and Related Work}
This section reviews relevant work in temporal logic and query checking and defines terminology for finite data streams.

\subsection{Temporal Logic}

\emph{Temporal logics} refer to a class of formalisms for reasoning about system behavior that evolves over time.
First introduced to the computer-science community by Amir Pnueli~\cite{pnueli:fcs1977}, temporal logics have been extensively studied as the basis for formal specification and verification of many types of hardware~\cite{eisner:2018} and software~\cite{holzmann:2004} systems.
These logics extend traditional propositional / first-order logics with additional operators, or \emph{modalities}, for describing when different system properties become, and remain, true.
When a system is finite-state, the task of checking whether or not it satisfies a given temporal property is usually decidable; \emph{model checkers}~\cite{baier:2008,clarke:2018} are tools used for this purpose.
Similarly, checking whether a temporal formula is \emph{satisfiable} (i.e.\/ a system exists that can satisfy it) is also in general decidable, in which case the temporal logic itself is often termed decidable~\cite{sistla:jacm1985}.
The rest of this subsections briefly reviews different temporal logics relevant for this paper.

\subsubsection{Linear Temporal Logic (LTL) and Finite Variants.}
Pnueli's original logic, and its variants, are called \emph{Linear Temporal Logic}, or LTL, because formulas are interpreted with respect to infinite sequences of system states.  These sequences can be thought of as non-terminating system executions; a system satisfies an LTL formula if all its executions do.  LTL typically extends propositional logic with a single modality $\tlU$, or ``until,'' from which other modalities such as $\tlF$ (``eventually'') and $\tlG$ (``always'') can be derived.  As an example, the LTL formula $\tlG (\textit{acq} \implies \tlF \textit{rel})$ holds of an execution if, whenever proposition \textit{acq} is true, signifying that a given resource has been acquired, then eventually proposition \textit{rel} becomes true, meaning that the resource has been released.  Model-checking and satisfiability checking are both decidable for LTL and typically rely on the use of automata-theoretic techniques~\cite{vardi:automata1986}.

In this paper we are concerned with LTL interpreted over \emph{finite}, rather than infinite, sequences of states.  Later we define the logic formally, which we call Finite LTL~\cite{huang:arxiv2019}.  Here we describe other accounts of LTL interpreted with respect to finite sequences.
De Giacomo and Vardi \cite{degiacomo:ijcai2013} give a version of finite-sequence LTL ($LTL_{f}$) 
and discuss how another logic, Linear Dynamic Logic over finite traces, $LDL_{f}$, is more expressive than $LTL_{f}$ but still has a PSPACE-complete satisfiability problem.
De Giacomo et al.~\cite{degiacomo:aaai2014} argue about how conversion from standard LTL over infinite sequences to finite sequences is often misused or misappropriated and give a definition close to $LTL_{f}$.
Ro{\c{s}}u \cite{rosu:rv2016} presents a sound and complete proof system for his version of finite-sequence temporal logic.
Fionda and Greco \cite{fionda:aaai2016} restrict negation to only atomic formulas and  consider a fragment of finite-sequence LTL ($LTL_{f,s}$) that allows only a single atomic proposition to be true at each time instance.  All of these logics require the sequences used to interpret formulas to be non-empty.  In contrast, our
Finite LTL, while syntactically identical to some of these logics~\cite{degiacomo:aaai2014,degiacomo:ijcai2013}, allows empty sequences as models also.  This choice is driven by algorithmic concerns.

We close our discussion of LTL with a discussion of finite interpretations of LTL used in runtime-monitoring applications~\cite{leucker:2017,rosu:rv2016}.  The interpretations of these logics reflect when the given finite sequence can be extended to an infinite sequence satisfying the LTL formula in the traditional sense.  As such, these semantic accounts return one of three truth values, given a finite sequence and a formula:  ``must satisfy'' (meaning every infinite extension of the finite formula satisfies the LTL formula); ``may satisfy'' (some infinite extension, but not all extensions, of the finite sequence satisfies the LTL formula); and ``does not satisfy'' (no infinite extension of the finite sequence satisfies the LTL formula).

\subsubsection{Computation Tree Logic (CTL).}

Computation Tree Logic~\cite{clarke:lop1981}, or CTL, extends the modalities of LTL given above with so-called \emph{path quantifiers} describing whether all, or some, infinite executions satisfy the given modality.  CTL formulas are interpreted with respect to \emph{computation trees} rather than individual executions, where each node is labeled with a (not necessarily distinct) system state and the edges indicate the tree reachable after the next execution state.  These trees are typically represented implicitly as Kripke structures~\cite{browne:tcs1988}, which represent system state spaces as graphs.  When these Kripke structures are finite-state, CTL model checking is decidable in time linear in the number of states~\cite{clarke:toplas1986}.

\subsection{Query Checking}
Temporal-logic query checking was first developed by Chan in the early 2000s \cite{chan:cav2000}.  His formulation of the problem was as follows:  given a Kripke structure~\cite{browne:tcs1988} describing the operational behavior of a system, and a Computation-Tree Logic (CTL)~\cite{clarke:lop1981} query, or formula containing placeholders for missing subformulas, compute propositional formulas that, when plugged into the placeholders, yield a formula satisfied by the Kripke structure.  Chan further identified a fragment of CTL, which he called CTL$^{v}$, whose placeholders have unique strongest solutions, and presented an efficient algorithm for computing these.
 Subsequent work has extended the classes of CTL queries that can be checked~\cite{bruns:lics2001} and explored its use in system understanding~\cite{gurfinkel:tse2003}. Others~\cite{chockler:hsvt2011,huang:avocs2017} have also considered LTL, rather than CTL, as the basis for queries that are solved with respect to Kripke structures.  Because of its relevance to the current work, we describe our results on LTL query checking~\cite{huang:avocs2017} in more detail.  Drawing inspiration from the
traditional automaton-based LTL model-checking approach as first presented by Vardi and Wolper
\cite{vardi:automata1986}, we define a construction for converting a LTL query $\phivar$, where $\var$ represents a placeholder, into a so-called
B\"uchi template automaton $B_{\neg \phivar}$.  These template automata are like traditional B\"uchi automata except that transitions are labeled by propositional formulas over atomic propositions; they may also involve $\var$.  The property enjoyed by $B_{\neg \phivar}$ is that for every propositional formula $\gamma$, $B_{\neg \phi[\gamma]}$, where $\phi[\gamma]$ is the formula obtained by replacing all instances of $\var$ by $\gamma$, accepts exactly the infinite state sequences that \emph{violate} $\phi[\gamma]$.  This template automaton is then composed with the B\"uchi automaton corresponding to the given Kripke structure, yielding another template automaton $B_c$.  Choices for $\var$ in this automaton can result in edge labels being logically
unsatisfiable, which is akin to deleting the transition from $B_{c}$.  Following a line of reasoning 
parallel to the initial Vardi and Wolper work, should a solution $\gamma$ for $\var$ be selected that makes the
language $L(B_{c})$ empty, the same solution $\gamma$ makes every execution of the Kripke structure satisfy $\phi[\gamma]$, and hence the Kripke structure as a whole satisfies $\phi[\gamma]$.
The problem of finding a solution $\gamma$ for $\var$ is therefore translated into the challenge of identifying an assignment $\gamma$ for $\var$ such that no accepting path exists through the composed automaton $B_{c}$.  That paper gave a singly exponential algorithm for solving this problem and reported on a proof-of-concept implementation.


\subsection{Data Streams}
We now give formal definitions for data streams.
\begin{definition}\label{def:data-stream}
Let  $\AP$ be a set of atomic propositions and $\mathbb{N}$ the set of natural numbers.  Then a \emph{data stream} over $\AP$ is a finite sequence
$$
(t_0, A_0) \ldots (t_{n-1}, A_{n-1}) \in (\mathbb{N} \times 2^{\AP})^*
$$
  such that $t_{i} \leq t_{j}$ holds for all $0 \leq i \leq j < n$.  We sometimes refer to $(t_i, A_i)$ in a data stream as an \emph{observation} and $t_i$ as the \emph{time stamp} of the observation. We use $\PiAP$ to represent the set of all data streams over $\AP$.
\end{definition}
Intuitively, atomic propositions are observations that can be made about the state of a system as it executes.  Data stream $\pi = (t_0, A_0) \ldots (t_{n-1}, A_{n-1})$ can then be seen as the result of observing the system for a finite period of time, where at each time instant $t_i$ the atomic propositions $A_i \subseteq \AP$ are true while those in $\AP / A_i$ are false. The condition imposed by Definition~\ref{def:data-stream} on time stamps requires that time advances monotonically throughout the data stream.  We use $\varepsilon$ to denote the empty data stream (the length-0 sequence).  We write $|\pi| = n$ for the length of data stream $\pi = (t_0, A_0) \ldots (t_{n-1}, A_{n-1})$, $\pi_i = (t_i, A_i)$ for the $i^{\textnormal{th}}$-indexed step in in $\pi$, and $\pi(i) = (t_i, A_i) \ldots (t_{n-1},A_{n-1})$ for the suffix of $\pi$ obtained by removing the first $i$ elements from the data stream.  Note that $\pi_i$ is only defined when $i < |\pi|$, while $\pi(i)$ is defined when $i \leq |\pi|$, and that $\pi(0) = \pi$ and $\pi(|\pi|) = \varepsilon$.

In the rest of the paper we will focus on so-called \emph{normalized} data streams, which are defined as follows.

\begin{definition}\label{def:normalized-data-stream}
Data stream $\pi$ over $AP$ is \emph{normalized} if and only if for all $i$ such that $0 \leq i < |\pi|$, $\pi_i$ has form $(i, A_i)$.
\end{definition}
In a normalized data stream, the time stamps of the elements in the sequence begin at 0 and increase by 1 at every step.  In such data streams we can omit the explicit time stamp and instead represent the stream as a finite sequence $A_0 \ldots A_{n-1}$.  For normalized data stream $\pi = A_0 \ldots A_{n-1}$ we abuse notation and write $\pi(i)$ as follows:  $\pi(i) = A_i \ldots A_{n-1}$.  This definition makes $\pi(i)$ normalized.\footnote{This detail, while necessary to point out, is not important in what follows, since the properties we consider in this paper are insensitive to specific time-stamp values.}

\section{Finite LTL Query Checking}


This section introduces the Finite LTL Query Checking problem over normalized data streams.  We begin by defining the logic Finite LTL that is the basis for our queries.  We then show how to define queries based on this logic.

\subsection{Finite LTL}

LTL is interpreted with respect to infinite sequences of states; in contrast, our data streams are finite, reflecting the intuition that an observation of a system execution must end at some point.  Accordingly, while Finite LTL has the same syntax as LTL, its semantics is different, as it is given in terms of normalized data streams.  In what follows, fix a (nonempty) set $\AP$ of atomic propositions.

\subsubsection{Syntax of Finite LTL.}
The set of Finite LTL formulas is defined by the following grammar, where $a \in \AP$.
$$
\phi ::=  a \mid \neg \phi \mid \phi_{1} \land \phi_{2} \mid \tlX \phi \mid \phi_{1} \tlU \phi_{2}
$$
Finite LTL formulas may be constructed from atomic propositions using the traditional propositional operators $\neg$ (``not'') and $\land$ (``and''), as well as the modalities of ``next'' ($\tlX$) and ``until'' ($\tlU$).  We use $\Phi^{\AP}$ to refer to the set of all Finite LTL formulas, omitting $\AP$ if it is clear from context.
We call formulas that do not involve any use of $\tlX$ or $\tlU$ \emph{propositional}, and write $\Gamma^{\AP} \subsetneq \Phi^\AP$ for the set of all such propositional formulas.  We also have the following derived notations.
\begin{minipage}{0.5\textwidth}
\begin{eqnarray*}
\logicfalse	&=& a \land \neg a \\
\logictrue 	&=& \neg \logicfalse\\
\phi_1 \lor \phi_2	&=& \neg ((\neg \phi_1) \land (\neg \phi_2))\\
\phi_1 \implies \phi_2 &=& (\neg \phi_1) \lor \phi_2
\end{eqnarray*}
\mbox{}
\end{minipage}
\begin{minipage}{0.48\textwidth}
\begin{eqnarray*}
\phi_1 \tlR \phi_2	&=& \neg ((\neg \phi_1) \tlU (\neg \phi_1))\\
\tlWeakX \phi		&=& \neg \tlX (\neg \phi)\\
\tlF \phi			&=& \mathit{true} \tlU \phi\\
\tlG \phi			&=& \neg \tlF (\neg \phi)
\end{eqnarray*}
\mbox{}
\end{minipage}

\noindent
The constants $\mathit{false}$ and $\mathit{true}$, and the operators $\land$ and $\lor$, and $\tlU$ and $\tlR$, are duals in the usual logical sense, with $\tlR$ sometimes referred to as the ``release'' operator.  We introduce $\tlWeakX$ (``weak next'') as the dual for $\tlX$.  That this operator is needed is due to our use of (finite) data streams to interpret Finite LTL formulas; this means that, in contrast to regular LTL, $\tlX$ is not its own dual.  This point is elaborated on later.  Finally, the duals $\tlF$ and $\tlG$ capture the usual notions of ``eventually'' and ``always'', respectively.

\subsubsection{Semantics of Finite LTL.}

The semantics of Finite LTL is given as a relation $\pi \sat \phi$ defining when normalized data stream $\pi$ satisfies formula $\phi$.

\begin{definition}\label{finite-ltl-semantics}
Let $\phi$ be a Finite LTL formula, and let $\pi$ be a normalized data stream.  Then $\pi \models \phi$ is defined inductively on the structure of $\phi$ as follows.
\begin{itemize}
\item $\pi \sat a$ iff $|\pi| \geq 1$ and $a \in \pi_0$
\item $\pi \sat \neg \phi$ iff $\pi \not\sat \phi$
\item $\pi \sat \phi_1 \land \phi_2$ iff $\pi \sat \phi_1$ and $\pi \sat \phi_2$
\item $\pi \sat \tlX \phi$ iff $|\pi| \geq 1$ and $\pi(1) \sat \phi$
\item $\pi \sat \phi_1 \tlU \phi_2$ iff $\exists j \colon 0 \leq j \leq |\pi| \colon \pi(j) \sat \phi_2$ and $\forall k \colon 0 \leq k < j \colon \pi(k) \sat \phi_1$
\end{itemize}
We write $\den{\phi}$ for the set $\{ \pi \mid \pi \sat \phi \}$ of all data streams satisfying $\phi$, and say that $\phi_1$ and $\phi_2$ are \emph{logically equivalent}, notation $\phi_1 \equiv \phi_2$, if $\den{\phi_1} = \den{\phi_2}$.  We say that $\phi_1$ is \emph{weaker than} $\phi_2$ (and $\phi_2$ is \emph{stronger than} $\phi_1$) if $\den{\phi_2} \subseteq \den{\phi_1}$, and write $\phi_1 \leq \phi_2$ in this case.  If $\phi_1 \leq \phi_2$ but $\phi_1 \not\equiv \phi_2$ we say $\phi_1$ is \emph{strictly weaker than} $\phi_2$ ($\phi_2$ is \emph{strictly stronger than} $\phi_1$).
\end{definition}

Intuitively, $\pi \models \phi$ holds if the data stream $\pi$ satisfies $\phi$.  In particular, data stream $\pi$ satisfies atomic proposition $a$ iff $|\pi| > 0$ (so $\pi_0$, the ``current state'', is defined) and $a \in \pi_0$, while $\pi$ satisfies $\neg \phi$ iff it fails to satisfy $\phi$. Satisfying $\phi_1 \land \phi_2$ requires satisfying both $\phi_1$ and $\phi_2$ individually.  For $\pi$ to satisfy $\tlX \phi$ it must be the case that $|\pi| > 0$, so that $\pi(0)$ exists; then $\pi$ makes $\tlX \phi$ true provided that $\pi(0)$, the sequence beginning in the ``next state'' of $\pi$ satisfies $\phi$.  Finally, $\pi$ satisfies $\phi_1 \tlU \phi_2$ iff there is some point in the sequence at which $\phi_2$ becomes true, and at all points in the sequence leading up to that point, $\phi_1$ holds.

The meaning of the derived operators can be understood from their definitions, but we do wish to comment on $\tlF$, $\tlG$ and $\tlWeakX$.  In particular, $\pi \sat \tlF \phi$ holds if iff at some point in $\pi$, $\phi$ is satisfied.  Similarly, $\pi \sat \tlG \phi$ holds iff at every point in $\pi$, $\phi$ is true.  As for $\tlWeakX \phi$, its semantics is nuanced:  $\pi$ can satisfy $\tlWeakX \phi$ if either $\pi = \varepsilon$
(because $\varepsilon \not\sat \tlX\lnot\phi'$ for any formula $\phi'$)
, or $|\pi| > 0$ and $\pi(1)$, the sequence beginning at the next state of $\pi$, satisfies $\phi$.  Put another way, for $\pi$ to satisfy $\tlWeakX \phi$ it can either be empty, or satisfy $\tlX \phi$.  The last observation highlights a difference between LTL and Finite LTL: $\tlX$, which is its own dual in LTL, does not have this property in Finite LTL, since $\varepsilon \not\sat \tlX \phi$ for all $\phi$.

The previous observations generally highlight subtleties in the semantics related to $\varepsilon$; for example, $\varepsilon \sat \neg a$ for any atomic proposition $a$, and the meanings of $\tlF \phi$ and $\tlG \phi$ can be non-intuitive depending on whether or not $\varepsilon \sat \phi$.  The paper~\cite{huang:arxiv2019} discusses the semantics in detail and also shows how issues relating to $\tlF$ and $\tlG$ in particular may be addressed via simple encodings.  It also gives a construction for building non-deterministic finite automata from formulas whose languages consist of the data streams satisfying the given formula.  We use key aspects of that translation in this paper, and so have elected to retain $\varepsilon$ as a sequence for interpreting Finite LTL formulas. 

\subsubsection{The propositional fragment of Finite LTL.}
This paper makes heavy use of $\Gamma^\AP$, the propositional fragment of Finite LTL, so we comment briefly on the properties of it here.  First, it is easy to see that for any propositional formula $\gamma \in \Gamma^\AP$ and non-empty data streams $\pi$ and $\pi'$ such that $\pi_0 = \pi'_0$, $\pi \sat \gamma$ iff $\pi' \sat \gamma$.  That is, only the first state in a non-empty data stream matters for determining whether or not the stream satisfies $\gamma$.  It is also shown in~\cite{huang:arxiv2019} that for any $\gamma \in \Gamma^\AP$, $\varepsilon \sat \gamma$ iff $\pi \sat \gamma$ for all non-empty $\pi$ such that $\pi_0 = \emptyset$ (i.e.\/ $\pi_0$ makes every atomic proposition false).  
This means that, despite the presence of $\varepsilon$ as a sequence, the propositional formulas $\Gamma^\AP$ enjoy the usual properties of propositional logic: deMorgan's Laws, distributivity, convertability into positive normal form and disjunctive normal form, etc.  We close this section by remarking on a complete lattice structure built on $\Gamma^\AP$ that is used heavily later.

\begin{definition}
Let $\gamma \in \Gamma^\AP$.  The \emph{equivalence class of $\gamma$ with respect to $\equiv$} is defined as
$$[\gamma]_\equiv = \{ \gamma' \in \Gamma^\AP \mid \gamma' \equiv \gamma \}.$$
We write 
$C_{\Gamma^\AP} = \{ [\gamma]_\equiv \mid \gamma \in \Gamma_\AP \}$
for the set of equivalence classes of $\Gamma_\AP$ and
extend $\land$ and $\lor$ and $\leq$ to $C_{\Gamma^\AP}$ as follows.
\begin{align*}
[\gamma_1]_\equiv \land [\gamma_2]_\equiv &= [\gamma_1 \land \gamma_2]_\equiv
\\
[\gamma_1]_\equiv \lor [\gamma_2]_\equiv &= [\gamma_1 \lor\gamma_2]_\equiv
\end{align*}
We also lift $\leq$ to $C_{\Gamma^\AP}$ in the obvious manner:  $[\gamma_1]_\equiv \leq [\gamma_2]_\equiv$ iff $\gamma_1 \leq \gamma_2$.
Finally, we define $[\gamma_1, \gamma_2] = \{ \gamma' \in \Gamma^{\AP} \mid \gamma_1 \leq \gamma' \leq \gamma_2\}$ to be the \emph{propositional interval} bounded below by $\gamma_1$ and above by $\gamma_2$, inclusive.
\end{definition}

\noindent
If $\AP$ is finite it follows that $C_{\Gamma^\AP}$ is finite, even though $\Gamma^\AP$ is not; indeed the number of equivalence classes in $C_{\Gamma^\AP}$ is doubly exponential in the size of $\AP$.  We  now state the following well-known result about $C_{\Gamma^\AP}$.

\begin{theorem}
Ordering $\leq$ induces a lattice on $C_{\Gamma^\AP}$, with $\land$ being the least upper bound (join) operator and $\lor$ being the greatest lower bound (meet) operator.
\end{theorem}

\noindent
It is easy to see that $[\logicfalse]_\equiv$ is the maximum element in the lattice, while the minimum one is $[\logictrue]_\equiv$.
In addition, $[\gamma_1]_\equiv \land [\gamma_2]_\equiv$ is the least upper bound of $[\gamma_1]_\equiv$ and $[\gamma_2]_\equiv$, while $[\gamma_1]_\equiv \lor [\gamma_2]_\equiv$ is the greatest lower bound of $[\gamma_1]_\equiv$ and $[\gamma_2]_\equiv$.

\subsection{Query Checking for Finite Data Streams}
In our work on LTL Query Checking \cite{huang:avocs2017}, we were interested in solving LTL \emph{queries} over Kripke structures.  In that setting a query is a LTL formula containing a missing propositional subformula, and the goal is to construct solutions for the missing subformula.  In this paper, we instead are interested in \emph{Finite} LTL queries and normalized data streams obtained by observing the behavior of the system in question.  This section defines this Finite LTL query-checking problem precisely and prove results used later in the paper.  
In what follows we restrict $\AP$ to be finite and non-empty.

Finite LTL queries correspond to Finite LTL formulas with a missing propositional subformula, which we denote $\var$.  It should be noted that $\var$ stands for an unknown \emph{propositional formula}; it is \emph{not} a (fresh) atomic proposition.  The syntax of queries is as follows:
$$
\phi := \var \mid a \in \AP
\mid \lnot \phi
\mid \phi_{1} \land \phi_{2}
\mid \tlX \phi
\mid \phi_{1} \tlU \phi_{2}
$$

In this paper we only consider the case of a single propositional unknown, although the definitions can naturally be extended to multiple such unknowns, as well as missing subformulas lifted to arbitrary Finite LTL formulas rather than only propositional formulas.  We often write $\phivar$ for an LTL query with unknown $\var$, and $\phi[\gamma]$ for the LTL formula obtained by replacing all occurrences of $\var$ by LTL propositional formula $\gamma$.  If $\gamma[\var]$ is a query containing no modalities then we call $\gamma[\var]$ a \emph{propositional query}.  We write $\Phi[\var]$ for the set of Finite LTL queries and $\Gamma[\var] \subsetneq \Phi[\var]$ for the set of propositional queries.  We also lift the notion of logical equivalence, $\equiv$, to LTL queries as follows:  $\phi_1[\var] \equiv \phi_2[\var]$ iff $\phi_1[\gamma] \equiv \phi_2[\gamma]$ for all $\gamma \in \Gamma$.


The query-checking problem $\qc(\Pi, \phivar)$ may be formulated as follows.

\begin{description}
  \item[Given:] Finite set $\Pi$ of normalized data streams, Finite LTL query $\phivar$
  \item[Compute:] All propositional $\gamma$ such that for all $\pi \in \Pi$, $\pi \sat \phi[\gamma]$
\end{description}

\noindent
If $\gamma$ is such that $\pi \sat \phi[\gamma]$ for all $\pi$ in $\Pi$, then we call $\gamma$ a \emph{solution} for $\Pi$ and $\phivar$, and in this case we say that $\phivar$ is solvable for $\Pi$.  Computing all solutions for $\qc(\Pi, \phivar)$ cannot be done explicitly, since the number of propositional formulas is infinite.  However, if $\AP$ is finite then we are able to give a finite representation of the solutions for $\qc(\Pi,\phivar)$ using equivalence classes of $\Gamma^\AP$.


As an example Finite LTL query, consider $\tlG \texttt{var}$.  A solution to this query would yield a formula that is invariant at every observation in every data stream in $\Pi$.  Another example of a finite LTL query $\phivar$ is $\tlG\left(\var \to \tlF \mathit{err}\right)$.  Assuming $\textit{err}$ is an atomic proposition representing the occurrence of an error condition, a solution  to this query would give conditions guaranteed to trigger a future system error.  Such information could be useful in subsequent root-cause analyses of why the error occurred.


\section{From Finite LTL Queries to Automata}\label{sec:ltltableau}

This section describes the basis for our query-checking approach for Finite LTL:  a mechanism for converting Finite LTL queries into automaton-based representations called \emph{finite query automata}.  The section following then gives algorithms for solving query automata for a given finite, non-empty set of normalized data streams.  Our method for generating query automata relies on the tableau-based approach in~\cite{huang:arxiv2019} for producing non-deterministic finite automata (NFAs) from Finite LTL formulas.  We first review the tableau-based approach for Finite LTL, then describe how it can be adapted to produce query automata. 

\subsection{From Finite LTL to NFAs}

The tableau construction in~\cite{huang:arxiv2019} works on Finite LTL formulas in \emph{positive normal form} (PNF).  A formula is in PNF if negation is only applied to atomic propositions; if Finite LTL is extended with the derived operators $\lor$, $\tlWeakX$ and $\tlR$ then any formula in this extended logic can be converted into PNF.  In what followswe assume Finite LTL is extended in this fashion.  We recall definitions associated with non-deterministic automata below.

\begin{definition}\label{def:nfa}
A \emph{non-deterministic finite automaton} (NFA) is a tuple $(Q, \Sigma, q_I, \delta, F)$, where:
\begin{itemize}
\item
$Q$ is a finite set of \emph{states};
\item\label{nfa-i}
$\Sigma$ is a finite non-empty set of \emph{alphabet} symbols;
\item
$q_I \in Q$ is the \emph{start state};
\item
$\delta \subseteq Q \times \Sigma \times Q$ is the \emph{transition relation}; and
\item
$F \subseteq Q$ is the set of \emph{accepting states}.
\end{itemize}
\item\label{nfa-ii}
Let $M = (Q, \Sigma, q_I, \delta, F)$ be a NFA, let $q \in Q$, and let $w \in \Sigma^*$.  Then $q$ \emph{accepts $w$ in $M$} iff one of the following hold.
\begin{itemize}
\item
$w = \varepsilon$ and $q \in F$
\item
$w = \sigma w'$ for some $\sigma \in \Sigma, w' \in \Sigma^*$ and there exists $(q,\sigma,q') \in \delta$ such that $q'$ accepts $w'$ in $M$.
\end{itemize}
$L(M)$, the \emph{language} of $M$, is
$
L(M) = \{ w \in \Sigma^* \mid q_I \textnormal{ accepts } w \textnormal{ in } M\}.
$
\end{definition}

The following theorem is proven in~\cite{huang:arxiv2019}, among other places.

\begin{theorem}\label{thm:FLTL-to-PNFA}
Let $\phi$ be a Finite LTL formula.  Then there exists a NFA $M_\phi = (Q, 2^\AP, q_I, \delta, F)$ such that $\den{\phi} = L(M_\phi)$.
\end{theorem}

\noindent
This theorem asserts that for any $\pi$, $\pi$ satisfies $\phi$ iff $\pi \in L(M_\phi)$, meaning that $M_\phi$ is a NFA-based characterization of the meaning of $\phi$.
The proof in~\cite{huang:arxiv2019} for the theorm uses a tableau-based construction~\cite{1985-WolperWOLTTM} to produce $M_\phi = (Q, 2^\AP, q_I, \delta, F)$ from $\phi$.  Since we rely on this construction, and the specific connections it makes between states in $M_\phi$ and subformulas of $\phi$, we give a brief overview of it.
We begin by noting that the alphabet of $M_\phi$ is defined to be $2^\AP$: each alphabet symbol corresponds to a subset of $\AP$.  It therefore follows that $L(M_\phi)$ consists of sequences of subsets of $\AP$, which are just normalized data steams. The state set $Q_\phi$ associates a unique set of subformulas of $\phi$ to each state; formally, $Q_\phi = 2^{S(\phi)}$, where $S(\phi)$ is the set of all subformulas of $\phi$.  Intuitively, the tableau construction guarantees that data streams processed starting from state $q$ in $M_\phi$ will be guaranteed to satisfy each formula associated with $q$.  We sometimes write $\bigwedge q$ for the formula obtained by  onjoining all the formulas associated with $q$, with $\bigwedge \emptyset$ taken to be $\mathit{true}$.  The start state $q_{I,\phi}$ is defined to be $\{ \phi \}$, i.e\/ the singleton set containing only $\phi$.  $F_\phi$, the set of accepting states, consists of states $q$ with the property that each subformula $\phi'$in $q$ is satisfied by $\varepsilon$ (i.e.\/ $\varepsilon \sat \bigwedge q$).  The transition relation $\delta$ is defined so that the following both hold.
\begin{itemize}
\item
If $(q_1, A, q_2) \in \delta$ for states $q_1, q_2$ and $A \subseteq \AP$, and data stream $\pi$ is such that $\pi \sat \bigwedge q_2$, then $A\pi \sat \phi_1$ for each $\phi_1 \in q_1$ ($A\pi$ is the data stream obtained by prefixing $\pi$ with $A$).
\item
If $A\pi \sat \bigwedge q_1$ then there exists $q_2$ such that $(q_1, A, q_2)$ and $\pi \sat \bigwedge q_2$.
\end{itemize}
These facts guarantee the desired correspondence between $\phi$ and $M_\phi$; see~\cite{huang:arxiv2019} for details.

That paper also discusses a symbolic representation for $M_\phi$ in which transitions, rather than labeled with subsets of $\AP$, are instead labeled with elements of $\Gamma^\AP$, the propositional Finite LTL formulas.  This representation is central to what follows, so we formalize it here.

\begin{definition}\label{def:pnfa}
A \emph{propositional NFA} (PNFA) is a tuple $(Q, \AP, q_I, \delta, F)$, where:
\begin{itemize}
\item
$Q$, $q_I$ and $F$ are as in Definition~\ref{def:nfa};
\item\label{nfa-i}
$\AP$ is a finite non-empty set of atomic propositions;
\item
$\delta \subseteq Q \times \Gamma^\AP \times Q$ is the \emph{transition relation}.
\end{itemize}
Let $M = (Q, \AP, q_I, \delta, F)$ be a PNFA, let $q \in Q$, and let $w \in (\Sigma_\AP)^*$, where $\Sigma_\AP = 2^\AP$.  Then $q$ \emph{accepts $w$ in $M$} iff:
\begin{itemize}
\item
$w = \varepsilon$ and $q \in F$; or
\item
$w = Aw'$ for some $A \subseteq \AP, w' \in (2^\AP)^*$, and there exists $(q,\gamma,q') \in \delta$ such that $A \sat \gamma$ and $q'$ accepts $w'$ in $M$.\footnote{Recall that $A \in 2^\AP$ is also a singleton data stream, and thus $A \sat \gamma$ is defined.}
\end{itemize}
$L(M)$ is defined as in Definition~\ref{def:nfa}.
\end{definition}

In a PNFA, 
transitions are labeled by propositional formulas $\gamma$ constructed from $\AP$; such a PNFA is intended to accept sequences of alphabet symbols $A$, where each $A \subseteq \AP$ is a set of atomic propositions.  To process such a sequence the machine begins in its start state, then consumes each symbol $A$ by selecting transitions emanating from the current state and checking if $A$ satisfies the propositional formula labeling the transition.  If $A$ satisfies the label of the transition it may be taken, with $A$ being consumed and the next state being updated to the target state of the transition.  If it is possible to reach an accepting state after processing the last symbol then the machine accepts the sequence; otherwise, it does not.

\subsubsection{Emptiness checking for PNFAs.}
Later in this paper we will need to check whether $L(M) = \emptyset$ for PNFA $M = (Q, \AP, q_I, \delta, F)$.  This obviously holds iff no accepting state $q' \in F$ is reachable via a sequence of \emph{live} transitions from start state $q_I$, where a transition $(q, \gamma, q')$ is live iff $\gamma$ is satisfiable and \emph{dead} otherwise.  Computing emptiness of $L(M)$ can be solved using standard reachability techniques on a graph derived from $M$ as follows:  the nodes of the graph are $Q$, and there is an edge $q \to q'$ in the graph iff there is a live transition $(q, \gamma, q') \in \delta$.  Then $L(M) = \emptyset$ iff the $Q' \cap F = \emptyset$, where $Q'$ is the set of nodes in the graph reachable from $q_I$.  From a complexity-theoretic point of view the liveness check in the construction of this graph is the most expensive operation, although it can be avoided if the PNFA is constructed in a way that guarantees liveness of transition labels, as is the case for example in the construction given in~\cite{huang:arxiv2019}.

\subsection{From Finite LTL Queries to FQAs}

Our query-checking methodology relies on converting Finite LTL queries into \emph{Finite Query Automata} (FQAs).
A FQA is like a PNFA except that the FQA's transition labels are propositional queries instead of formulas, and the acceptance condition depends on the propositional unknown embedded in the queries.

\begin{definition}
Let $\var$ be proposition variable.  A \emph{Finite Query Automaton (FQA)} $M[\var]$ is a tuple $(Q, \AP, q_I, \delta[\var], F[\var])$, where:
  \begin{itemize}
    \item $Q$ is a finite set of states;
    \item $\AP$ is a finite, non-empty set set of atomic proposition;
    \item $q_{I} \in Q$ is the initial state;
    \item $\delta[\var] \subseteq Q \times \Gamma[\var] \times Q$ is the transition relation; 
    \item $F[\var] \in Q \rightarrow \Gamma[\var]$ is the \emph{acceptance condition}.
  \end{itemize}
If $\gamma \in \Gamma$ is a propositional formula then we write $M[\gamma]$, the \emph{instantiation} of $M[\var]$ with $\gamma$, for the PNFA $(Q, \AP, q_I, \delta[\gamma], F[\gamma])$, where
$$
\delta[\gamma] = \{ (q, \gamma'[\gamma], q') \mid (q, \gamma'[\var], q') \in \delta[\var] \}.
$$
and $F[\gamma] = \{ q \in Q \mid \varepsilon \models (F[\var](q))[\gamma] \}$.
\end{definition}

An FQA is intended to be the automaton analog of a Finite LTL query, where $\var$ is the unknown proposition to be solved for.  An instantiation of an FQA with $\gamma$ is then the PNFA obtained by replacing $\var$ by $\gamma$ in the transition labels and in the queries $F[\var]$ associates with each state $q$.  In the latter case, the resulting instantiated queries are used to determine if the associated state is accepting or not:  it is accepting iff the empty stream $\varepsilon$ satisfies its instantiated query.  
Note that $\gamma$ can have two effects on the language of $M[\gamma]$:  one via  the transition relation, as some transitions may become dead, and the other via the accepting / non-accepting status of states.

Our method for query-solving is automaton-theoretic; it is based on constructing a FQA $M_{\phivar}[\var]$ from a Finite LTL query $\phivar$.  Our method for computing $M_{\phivar}[\var]$ uses a modification of the tableau construction in~\cite{huang:arxiv2019}; we sketch the idea here.  It can be shown that any query $\phivar$ can be put into PNF, where $\neg$ can only be applied to atomic propositions or instances of $\var$.  
We may then define $M_{\phivar}[\var]$ to be $(Q, \AP, \delta[\var], \phivar, F[\var])$, where each $q \in Q$ is a query $\phi_q[\var]$ consisting of the conjunction of a set of subqueries of $\phivar$; each $(q, \gamma[\var], q') \in \delta[\var]$ is a propositional query based on the tableau construction, and
$F[\var](q)$ is a propositional query $\gamma_q[\var]$ with the property that for all $\gamma \in \Gamma$, $\varepsilon \models \phi_q[\gamma]$ iff $\varepsilon \models \gamma_q[\gamma]$.\footnote{That such a $\gamma_q[\var]$ exists is a consequence of the fact that, as shown in~\cite{huang:arxiv2019}, checking whether or not $\varepsilon \models \phi$ for $\phi \in \Phi$ can be computing via induction on $\phi$:  in effect, the modal operators can be ignored.}
We have the following.

\begin{theorem}\label{thm:FQA}
Let $\phivar$ be a PNF Finite LTL query, and let $\gamma$ be a propositional formula.  Then $L(M_{\phivar}[\gamma]) = \den{\phi[\gamma]}$.
\end{theorem}


\noindent
Later in this paper we are especially interested in propositions $\gamma$ such that $L(M[\gamma]) = \emptyset$.
\begin{definition}
Let $M[\var]$ be a FQA.  Then $\gamma$ is a \emph{shattering condition} for $M[\var]$ for if $L(M[\gamma]) = \emptyset$.  If a shattering condition $\gamma$ exists for $M[\var]$ we say that $M[\var]$ is \emph{shatterable}.
\end{definition}

\noindent
We close with a discussion of the acceptance condition of FQA $M[\var] = (Q, \AP, q_I, \delta[\var], F[\var])$.  Recall that $F[\var] \in Q \rightarrow \Gamma[\var]$; that is, $F[\var]$ maps each $q \in Q$ to a propositional query $\gamma_q[\var]$.  If $\gamma \in \Gamma$ is subsequently used to instantiate $\var$, then $q$ is accepting iff $\gamma_q[\gamma]$ is satisfied by $\varepsilon$.  The intuition behind this definition is that in $M[\var]$ propositional queries are used as transition labels, and thus govern when transitions may be taken in its instantiations, and as acceptance criteria.  In a PNFA, a string is accepted iff it is possible to process the entire string and reach an accepting state:  this means that the remainder of the string to process is empty when the accepting state is entered.  We adopt this same convention in the definition of an instantiation of an FQA; if the instantiated acceptance query for state is satisfied by $\varepsilon$, and a data stream has been fully processed, leaving the instantiation of the FQA in this state, then the stream should be accepted, and the state should be accepting.

Somewhat surprisingly, this definition implies that any instantiation $M[\gamma]$ of $M[\var]$ can only have one of two possible sets of accepting states.  To see why, define an equivalence relation $\sim_\varepsilon \,\subseteq \Gamma \times \Gamma$ as follows:  $\gamma_1 \sim_\varepsilon \gamma_2$ if it is the case that $\varepsilon \models \gamma_1$ iff $\varepsilon \models \gamma_2$.  It is easy to see that $\sim_\varepsilon$ induces two equivalence classes on $\Gamma$:  $[\mathit{true}]_{\sim_\varepsilon}$, consisting of $\gamma$ such that $\varepsilon \models \gamma$, and $[\mathit{false}]_{\sim_\varepsilon}$, consisting of $\gamma'$ such that $\varepsilon \not\models \gamma'$. These equivalence classes have the following interval characterization.

\begin{lemma}\label{lem:epsilon}
\mbox{}
\begin{enumerate}
\item $[\mathit{true}]_{\sim_\varepsilon} = [\mathit{true}, \bigwedge_{a \in AP} \neg a]$.
\item $[\mathit{false}]_{\sim_\varepsilon} = [\bigvee_{a \in \AP} a, \mathit{false}]$.
\end{enumerate}
\end{lemma}
\begin{proof}
Follows from the facts that $\bigwedge_{a \in AP} \neg a$ is the strongest propositional formula satisfied by $\varepsilon$ and  that $\bigvee_{a \in \AP} a$ is the weakest propositional formula not satisfied by $\varepsilon$.
\end{proof}

\noindent
We now have the following.

\begin{lemma}\label{lem:eps}
Let $\gamma_1, \gamma_2 \in \Gamma$.
\begin{enumerate}
\item
If $\gamma_1 \equiv \gamma_2$ then $\gamma_1 \sim_\varepsilon \gamma_2$.
\item
If $\gamma_1 \sim_\varepsilon \gamma_2$ then for any Finite LTL query $\phi[\var]$, $\varepsilon \models \phi[\gamma_1]$ iff $\varepsilon \models \phi[\gamma_2]$.
\end{enumerate}
\end{lemma}
\begin{proof}
  (1) is immediate.  (2) relies on the fact that determining if $\varepsilon \models \phi$ for Finite LTL formula $\phi$ can be computed inductively on the structure of $\phi$.
\end{proof}

\noindent
From this lemma, we observe that for any $\gamma \in \Gamma$,  $M[\gamma]$ can have only one of two possible sets of accepting states:  $F[\logictrue]$, when $\varepsilon \models \gamma$, or $F[\logicfalse]$, when $\varepsilon \not\models \gamma$.

\subsection{Composing PNFAs and FQAs}

We close this section by adapting the well-known language-intersection composition operation, $\otimes$, to PNFAs and FQAs.

\begin{definition}\label{def:composition}
Let $M_i$, $i \in \{1, 2\}$, be PNFAs $(Q_i, \AP, q_i, \delta_i, F_i)$.  Then $M_1 \otimes M_2$ is PNFA $(Q_1 \times Q_2, \AP, (q_1, q_2), \delta_{1,2}, F_1 \times F_2)$ where:
$$
\delta_{1,2} =
\{
((q_1', q_2'), \gamma_1 \land \gamma_2, (q_1'', q_2''))
\mid
(q_1', \gamma_1, q_1'') \in \delta_1 \textnormal{ and }
(q_2', \gamma_2, q_2'') \in \delta_2.
$$

\end{definition}

\noindent
Operation $\otimes$ can be extended to the case when one of the $M_i$ is a FQA in an obvious manner.  Without loss of generality assume $M_1$ is PNFA $(Q_1, \AP, q_1, \delta_1, F_1)$ and let $M_2[\var]$ be the FQA $(Q_2, \AP, q_2, \delta_2[\var], F_2[\var])$.  Then $(M_1 \otimes M_2)[\var]$ is the FQA $(Q_1 \times Q_2, \AP, (q_1, q_2), \delta_{1,2}[\var], F_{1,2}[\var]$), where $\delta_{1,2}[\var]$ is defined as $\delta_{1,2}$ in Definition~\ref{def:composition} and
$
F_{1,2}[\var](\gamma) = F_1 \times (F_2[\gamma]).
$
We have the following.

\begin{theorem}\label{thm:composition}
Let $M_1$ be a PNFA.
\begin{enumerate}
\item If $M_2$ is a PNFA then $L(M_1 \otimes M_2) = L(M_1) \cap L(M_2)$.
\item If $M_2[\var]$ is a FQA then for every $\gamma \in \Gamma$, $L((M_1 \otimes M_2)[\gamma]) = L(M_1) \cap L(M_2[\gamma])$.
\end{enumerate}
\end{theorem}

\noindent Point (1) states the usual result that the language of the composed PNFA automata is equal to the intersection of the individual languages.  Point (2) establishes that if one of the automaton is instead a QFA, then for all propositional formulas $\gamma$, the language upon instantiating the composed PNFA with $\gamma$ is the same as the intersection of the language of the PNFA with the language of the QFA, instantiated by $\gamma$.

\section{Shattering FQAs}

The basis for our query-checking procedure is the computation of shattering conditions for specially constructed FQAs.  In this section we give a procedure for computing these conditions for general FQAs.  The algorithm relies on computing \emph{shattering intervals} for propositional queries $\gamma[\var]$.  We first show how this is done, then present our general FQA-shattering approach.

\subsection{Shattering Propositional Queries}

Our approach to shattering $M[\var]$ relies on selecting $\gamma$ so that some transitions in $M[\gamma]$ become dead because their labels, which are instantiated propositional queries, are unsatisfiable.
If the combination of acceptance sets and disabled transitions is such that no accepting state in $M[\gamma]$ is reachable, then $\gamma$ shatters $M[\var]$.
A key operation is the computation of all $\gamma' \in \Gamma$ for a given propositional query $\gamma[\var]$ such that $\gamma[\gamma']$ is unsatisfiable.
We call such a $\gamma'$ a \emph{shattering condition} for $\gamma[\var]$.
If $\gamma[\var]$ indeed has a shattering condition (it might not) we call $\gamma[\var]$ \emph{shatterable}.
In this section we show that the shattering conditions for shatterable $\gamma[\var]$ can be represented as a propositional interval $[\gamma_1, \gamma_2]$ and show how to compute this interval.
We start by considering special cases of $\gamma[\var]$.
We say that $\var$ is \emph{positive} in propositional query $\gamma[\var]$ iff every occurrence of $\var$ is within the scope of an even number of negations and \emph{negative} iff every occurrence is within the scope of an odd number of negations.  
These notions lead immediately to the following results.

\begin{theorem}\label{thm:shatterability}
  Let $\gamma[\var] \in \Gamma[\var]$ be a propositional query.
\begin{enumerate}
\item If $\var$ is positive in $\gamma[\var]$ then $\gamma[\var]$ is shatterable if and only if $\gamma[\logicfalse]$ is unsatisfiable.
\item If $\var$ is negative in $\gamma[\var]$ then $\gamma[\var]$ is shatterable if and only if $\gamma[\logictrue]$ is unsatisfiable.
\end{enumerate}
\end{theorem}
\begin{proof}
  For case 1, first assume that $\var$ is positive and $\gamma[\logicfalse]$ is unsatisfiable.
  Then, by definition, $\gamma[\var]$ is shatterable because $\logicfalse$ is a shattering condition for $\gamma[\var]$.
  Next, assume that $\var$ is positive and $\gamma[\var]$ is shatterable.
  This means that there is a $\gamma' \in \Gamma$ such that $\gamma[\gamma']$ is unsatisfiable.
  The proof follows from the fact that when $\var$ is positive in $\gamma[\var]$ and $\gamma_{1}' \le \gamma_{2}'$, then $\gamma[\gamma_{1}'] \le \gamma[\gamma_{2}']$.
  In particular, $\gamma' \le \logicfalse$ for all $\gamma' \in \Gamma$, so $\gamma[\gamma'] \le \gamma[\logicfalse]$.
  Thus, because $\gamma[\gamma']$ is unsatisfiable, $\gamma[\logicfalse]$ must also be unsatisfiable.
Case 2 is dual to case 1, and the proof follows similarly.
\end{proof}

\begin{theorem}\label{lem:pos-neg}
  Let $\gamma[\var] \in \Gamma[\var]$ be a shatterable propositional query.
  \begin{enumerate}
    \item\label{case1} If $\var$ is positive in $\gamma[\var]$, then there is a unique (modulo $\equiv$) weakest shattering formula $\gamma'$ for $\gamma[\var]$.
    \item\label{case2} If $\var$ is negative in $\gamma[\var]$, then there is a unique (modulo $\equiv$) strongest shattering formula $\gamma'$ for $\gamma[\var]$.
  \end{enumerate}
\end{theorem}
\begin{proof}
  Case~\ref{case1} follows from the fact that if $\var$ is positive in $\gamma[\var]$ and $\gamma[\var]$ is shatterable then $\gamma[\var]$ can be rewritten as $\var \land \gamma'$, where $\gamma[\var] \equiv \var \land \gamma'$ and $\gamma' \in \Gamma$ is a propositional formula.  It is straightforward to show that $\neg \gamma'$ shatters $\gamma[\var]$, and that if $\gamma''$ also shatters $\gamma[\var]$ then $\gamma'$ is stronger than $\neg \gamma'$ (i.e.\/ $\neg\gamma' \leq \gamma''$).  Case~\ref{case2} is dual; its proof is omitted.
\end{proof}

\noindent Characterizing the shatterability of general $\gamma[\var]$, in which $\var$ may appear both positively and negatively, is more complex and relies on the following.

\begin{definition}
Propositional query $\gamma[\var]$ is in \emph{shattering normal form (SNF)} iff it has form $\gamma_1 \lor (\var \land \gamma_2) \lor ((\neg \var) \land \gamma_3)$, where each $\gamma_i \in \Gamma$.
\end{definition}


\begin{lemma}\label{lem:snf}\mbox{}
For every $\gamma[\var] \in \Gamma[\var]$ there is a $\gamma'[\var] \in \Gamma[\var]$ in SNF such that $\gamma[\var] \equiv \gamma'[\var]$.
\end{lemma}
\begin{proof}
Observe that if we view $\var$ as an atomic proposition then we can also treat $\gamma[\var]$ as a propositional formula and convert it into \emph{disjunctive normal form}
in such a way that every clause has either no occurrences of $\var$, or one occurrence of $\var$, or one occurrence $\neg \var$.  We finish building $\gamma'[\var]$ by grouping the clauses containing $\var$ and then factoring out $\var$, and similarly for $\neg\var$.
\end{proof}


\begin{theorem}\label{thm:propositional-shatterability}
Let $\gamma[\var] = \gamma_1 \lor (\var \land \gamma_2) \lor ((\neg \var) \land \gamma_3)$ be in SNF.
\begin{enumerate}
\item
$\gamma[\var]$ is shatterable iff $\gamma_1$ is unsatisfiable and $\neg\gamma_2 \leq\gamma_3$.
\item
If $\gamma[\var]$ is shatterable then $\gamma'$ shatters $\gamma[\var]$ iff $\neg\gamma_2 \leq \gamma' \leq \gamma_3$.
\end{enumerate}
\end{theorem}
\begin{proof}
Follows from the definition of shatterability and that fact that if $\neg\gamma_2 \leq \gamma' \leq \gamma_3$ then $\gamma'$ shatters $(\var \land \gamma_2) \lor ((\neg \var) \land \gamma_3)$, and conversely.
\end{proof}

\noindent As a consequence of this theorem and Lemma~\ref{lem:snf}, we have that the set of shattering formulas for any shatterable propositional query $\gamma[\var]$ can be represented as a propositional interval $[\gamma_1, \gamma_2]$, which we refer to as the \emph{shattering interval} for $\gamma[\var]$.
Also note that when $\var$ is positive in $\gamma[\var]$, this interval has form $[\gamma', \mathit{false}]$, while if $\var$ is negative then the interval has form $[\mathit{true}, \gamma']$; here the $\gamma'$ are the shattering formulas guaranteed by Theorem~\ref{lem:pos-neg}.

We finally note a \emph{joint shattering interval}, which shatters each $\gamma_i[\var]$, can be computed for a finite set $\{ \gamma_1[\var], \ldots, \gamma_n[\var] \}$ of shatterable propositional queries, based on the following lemma.

\begin{lemma}[Based on \cite{ahlswede:ejc1996}]\label{lem:interval-conjunction}
Let $[\gamma_1', \gamma_1'']$ and $[\gamma_2', \gamma_2'']$ be shattering intervals for $\gamma_1[var]$ and $\gamma_2[\var]$, respectively.  Then $[\gamma_1' \land \gamma_2', \gamma_1'' \lor \gamma_2'']$ is the shattering interval for query $\gamma_1[\var] \land \gamma_2[\var]$.
\end{lemma}
\begin{proof}
Follows from the fact that $[\gamma_1', \gamma_1''] \cap [\gamma_2', \gamma_2''] = [\gamma_1' \land \gamma_2', \gamma_1'' \lor \gamma_2'']$.  Note that if this interval is empty then $\gamma_1[\var] \land \gamma_2[\var]$ is not shatterable.
\end{proof}

\noindent
In what follows we write $[\gamma_1', \gamma_1''] \land [\gamma_2', \gamma_2'']$ for $[\gamma_1' \land \gamma_2', \gamma_1'' \lor \gamma_2'']$, and $\bigwedge I = [\bigwedge \gamma_i', \bigvee \gamma_i'']$ when $I = \{[\gamma_1', \gamma_1''], \ldots, [\gamma_n', \gamma_n'']\}$ is a finite set of shattering intervals.

\subsection{Computing Shattering Conditions for FQAs}

We now describe an algorithm for computing a representation of the shattering conditions for FQA $M[\var]$.  The representation consists of a set of shattering intervals $\{ [\gamma_1, \gamma_1'], \ldots [\gamma_n, \gamma_n'] \}$ with the property that $\gamma$ shatters $M[\var]$ iff $\gamma \in [\gamma_i, \gamma_i']$ for some $1 \leq i \leq n$.

The algorithm works by collecting all transition labels in $M[\var]$, computing shattering conditions of each subset of the set of labels, then performing an emptiness check on the PNFAs (there may be two, since there are in general two possible sets of accepting states) obtained by removing the transitions whose labels fall in the considered set and treating other transitions as live.
For efficiency reasons, the algorithm also includes a preprocessing step, in which it computes the shattering interval of each transition in label $M[\var]$ and uses this information to eliminate transition labels, and transitions, that cannot affect subsequent emptiness checks.  Specifically, it replaces the labels of transitions that are unshatterable (i.e. have empty shattering intervals) by $\logictrue$, and it eliminates dead transitions (i.e.\/ transitions of form $(q, \gamma[\var], q')$ in which $\gamma[\var]$ has shattering interval $[\logictrue, \logicfalse]$) that are shattered by every propositional formula.  Pseudocode for the procedure may be found in Algorithm~\ref{alg:ShatterFQA}.

\begin{algorithm}[t]
  \SetKwInOut{Input}{Input}
  \SetKwInOut{Output}{Output}
  \SetKwBlock{Preprocessing}{Preprocessing step:}{} 
  \SetKwBlock{Main}{Main:}{} 
  \DontPrintSemicolon 

  \underline{\textbf{Algorithm} \textit{ShatterFQA}} $(M[\var])$\;
  \Input{FQA $M[\var] = (Q, \AP, q_I, \delta[\var], F[\var])$}
  \Output{Set $\mathit{SC}$ of shattering intervals}
  
  \Preprocessing{
    $\mathit{TL} := \{ \gamma[\var] \mid \exists q, q' \in Q .\, (q, \gamma[\var], q') \in \delta[\var] \}$
    	\tcp*{Transition labels}
    \lForEach{$\gamma[\var] \in \mathit{TL}$}{compute shattering interval $I_{\gamma[\var]}$ of $\gamma[\var]$}
    $\mathit{UL} := \{\gamma[\var] \in TL \mid I_{\gamma[\var]} = \emptyset \}$
    	\tcp*{Unshatterable labels}
    $\mathit{DL} := \{\gamma[\var] \in TL \mid I_{\gamma[\var]} = [\logictrue,\logicfalse]\}$
    	\tcp*{Dead labels}
    $\mathit{SL} := \mathit{TL} - (\mathit{UL} \cup \mathit{DL})$
    	\tcp*{Surviving labels}
    $\delta'[\var] := \{ (q, \gamma[\var], q') \in \delta[\var] \mid \gamma[\var] \in \mathit{SL} \}$\;
    $\delta'[\var] := \delta'[\var] \cup \{ (q, \logictrue, q') \mid \exists \gamma[\var] \in \mathit{UL}. \, (q, \gamma[\var], q') \in \delta[\var]\}$\;
    $M'[\var] := (Q, \AP, q_I, \delta'[\var], F[\var])$
    	\tcp*{$L(M'[\gamma]) = \emptyset$ iff $L(M[\gamma]) = \emptyset$}
  } 
  \Main{
    $\mathit{SC} := \emptyset$
    	\tcp*{Shattering conditions for $M[\var]$}
    \ForEach{$S \subseteq SL$}{
      $I_S := \bigwedge \{I_{\gamma[\var]} \mid \gamma[\var] \in S\}$\;\label{alg-line:I_S}
      $\mathit{PSC} := \emptyset$\;
      $\delta'' = \{ (q, \logictrue, q') \mid (q, \gamma[\var], q') \in \delta'[\var] \textnormal{ and } \gamma[\var] \not\in S\}$\;
      
      \uIf{$I_S \cap [\logictrue]_{\sim \varepsilon} \neq \emptyset$\label{alg-line:if-true}}{
        $M_t = (Q, \AP, q_I, \delta'', F[\logictrue])$
          \tcp*{PNFA with acc. set $F[\logictrue]$}
        \lIf{$L(M_t) = \emptyset$}{$\mathit{PSC} := \{ I_S \cap [\logictrue]_{\sim \varepsilon} \}$}
      } 
      
      \uIf{$I_S \cap [\logicfalse]_{\sim \varepsilon} \neq \emptyset$\label{alg-line:if-false}}{
        $M_f = (Q, \AP, q_I, \delta'', F[\logicfalse])$
          \tcp*{PNFA with acc. set $F[\logicfalse]$}
        \lIf{$L(M_f) = \emptyset$}{$\mathit{PSC} := \mathit{PSC} \cup \{ I_S \cap [\logictrue]_{\sim \varepsilon} \}$}
      } 
      
      \lIf(\tcp*[f]{Don't split $I_S$}){$|\mathit{PSC}| = 2$}{$\mathit{SC} := \mathit{SC} \cup \{I_S\}$}
      \lElse{$\mathit{SC} := \mathit{SC} \cup \mathit{PSC}$}
    } 
    \Return $(\mathit{S})$\;
  } 
\caption{Algorithm for shattering FQAs}\label{alg:ShatterFQA}

\end{algorithm}

The following establishes the correctness of this algorithm.

\begin{theorem}\label{thm:correctness}
Suppose $\mathit{SC} = \mathit{ShatterFQA}(M[\var])$.  Then
$\gamma \in \Gamma^\AP$ shatters $M[\var]$ iff $\gamma \in [\gamma', \gamma'']$ for some $[\gamma',\gamma''] \in SC$.
\end{theorem}

\begin{proof}
It can be seen that $M'[\var]$, as computed in Algorithm~\ref{alg:ShatterFQA}, has the property that for any $\gamma$, $L(M'[\gamma]) = \emptyset$ iff $L(M[\gamma]) = \emptyset$.  It is also the case that for any $S \subseteq \mathit{SL}$ in the algorithm, if $L(M_t) = \emptyset$ then $L(M'[\gamma]) = \emptyset$ for any $\gamma \in I_S \cap [\logictrue]_{\sim\epsilon}$, and similalrly for $M_f$ and $I_S \cap [\logicfalse_{\sim\epsilon}$.  Finally, we note that if $S, I_S$ and $\gamma \in I_S \cap [\logictrue]_{\sim\epsilon}$ are such that $L(M[\gamma]) = \emptyset$ but $L(M_t) \neq \emptyset$, then there is an $S' \supseteq S$ such that $\gamma \in I_{S'} \cap [\logictrue]_{\sim\epsilon}$ and $L(M_t) = \emptyset$.  A similar result holds for $M_f$, and therefore the theorem follows.
\end{proof}

We now analyze the complexity of Algorithm \textit{ShatterFQA} in terms of the number of conversions to DNF and propositional satisfiability checks required.  We first note that for a given propositional query $\gamma[\var]$, computing its shattering interval and shatterability requires one DNF conversion (to convert $\gamma[\var]$ into SNF) and two satisfiability checks; thus $O(1)$ DNF conversions / satisfiability checks suffice.  Thus, the preprocessing step of the algorithm requires $O(|\delta[\var]|)$ of these operations.  Now consider the \textbf{foreach} loop in the main processing step.  Each iteration requires two satisfiability checks to perform the tests in the first two \textbf{if} statements.  Thus the total number of such operations required is $O(2^{|\mathit{SL}|})$, and as $|\mathit{SL}| \leq |\delta[\var]|$, we get a complexity of $O(2^{|\delta[\var]|})$ DNF conversions and satisfiability checks.  (There are also $O(2^{|\delta[\var]|})$ PNFA-emptiness checks.)

This complexity is quite high, but we note that there are heuristics that can greatly reduce running time in practice.  First, we note that if $S \subseteq \mathit{SL}$ is such that $I_S$ shatters $M[\var]$, then so too does $I_{S'}$ for every $S \subseteq S'$.  This implies that a systematic enumeration of the subsets of $\mathit{SL}$ exploiting this fact can reduce computation time.  Secondly, if the constructions of $M_t$ and $M_f$ can exploit similar properties in order to enable incremental updates to the reachability analysis required for emptiness checking.  Finally, the structure of transition labels in $M[\var]$ can sometimes enable simple satisfiability checks; we explore this later.

\section{Solving $QC(\Pi, \phivar)$}
In this section we show how to solve the Finite LTL query-checking problem over normalized data streams using the FQA shattering framework developed in the previous section.  We consider both the single-stream and multiple-stream cases, and note how the structure of the specific shattering problems they induce enable significant optimizations.

\subsection{Query Checking Single Data Streams}

In the single-stream query-checking problem $QC1(\pi, \phivar) = QC(\{\pi\}, \phivar)$, we are given one normalized data stream $\pi$ and a Finite LTL query $\phivar$ and asked to compute all propositional formulas $\gamma \in \Gamma$ such that $\pi \models \phi[\gamma]$.
We can convert this problem into an FQA shattering problem as follows.
To begin with, it is obviously the case that $\pi \models \phi$ for any Finite LTL formula $\phi$ if and only if $\pi \not\models \lnot\phi$.
Now suppose $M_\pi$ is a PNFA such that $L(M_\pi) = \{\pi\}$.
It immediately follows that $\pi \models \phi$ if and only if $L(M_\pi) \cap L(M_{\lnot\phi}) = \emptyset$, where $M_{\lnot\phi}$ is a PNFA such that $\den{\lnot\phi} = L(M_{\lnot\phi})$ (see Theorem~\ref{thm:FLTL-to-PNFA}).  Theorem~\ref{thm:composition} ensures that this in turn holds if and only if $L(M_\pi \otimes M_{\lnot\phi}) = \emptyset$.  The same theorem, together with Theorem~\ref{thm:FQA}, then guarantees that for query $\phivar$ and a given propositional formula $\gamma$, $\pi \models \phi[\gamma]$ if and only if $L(M_\pi \otimes M_{\lnot(\phi[\gamma])}) = \emptyset$, or in other words, that $\gamma$ shatters FQA $(M_\pi \otimes M_{\lnot\phivar})[\var]$, where the construction of $M_{\lnot\phivar}$ is given in \secref{sec:ltltableau} immediately before Theorem~\ref{thm:FQA}.  Note that if if $\var$ is positive (negative) in $\phivar$, then $\var$ is negative (positive) in $M_{\lnot\phivar}[\var]$ and hence in $\left(M_{\pi} \otimes M_{\lnot\phivar}\right)[\var]$.

The above method requires constructing $M_{\pi}$; this is easily done as follows.  Let $\pi = A_0 \ldots A_{n-1}$, where each $A_i \subseteq \AP$.  Then $M_{\pi}$ consists of states ${q_0, \ldots, q_n}$, with each $q_i$ except $q_n$ having a single outgoing transition to $q_{i+1}$ labeled by a proposition $\dia{A_i} \in \Gamma^{\AP}$ derived from $A_i$ so that in a precise sense, $\dia{A_i}$ completely characterizes $A_i$.   Formally, if $A \subseteq \AP$ then let $\dia{A} = \bigwedge \{a : a \in A\} \land \bigwedge\{\lnot b : b \in \AP, b \not\in A\}$.   Note that for $A, A' \subseteq \AP$, $A' \models \dia{A}$ if and only if $A' = A$.
  $M_{\pi}$ has the form $(Q_\pi, \AP, q_{0}, \delta_{\pi}, \{q_{n}\})$, where $Q_{\pi} = \{q_{0}, \ldots, q_{n}\}$ and $\delta_{\pi} = \{(q_{i}, \dia{A_{i}}, q_{i+1}) : 0 \leq i < n\}$.

We may now give our method, also called $QC1(\pi, \phivar)$, for solving query problem $QC1(\pi, \phivar)$ as follows.
\begin{enumerate}
  \item Compute $M_{\pi}$
  \item Compute $M_{\lnot\phivar}$
  \item Compute the composed automaton $M_{c}[\var] = M_{\pi} \otimes M_{\lnot\phivar}$
  \item Return the result, $SC$, of invoking $\mathit{ShatterFQA(M_c[\var])}$
\end{enumerate}

\noindent
The next theorem establishes the correctness of $QC1$.

\begin{theorem}\label{thm:correctness-qc1}
Let $\pi$ be a normalized data stream, let $\phivar$ be a Finite LTL query, and let $SC$ be the set of shattering intervals returned by $QC1(\pi, \phivar)$.  Then for any propositional $\gamma$, $\pi \models \phi[\gamma]$ if and only if there exists $[\gamma_1, \gamma_2] \in SC$ such that $\gamma_1 \leq \gamma \leq \gamma_2$.
\end{theorem}
\begin{proof}
Follows from the discussion above.  Note that $\gamma_1 \leq \gamma \leq \gamma_2$ for some $[\gamma_1, \gamma_2] \in SC$ if and only if $\gamma$ shatters $M_{c}[\var]$.
\end{proof}

\subsubsection{Optimizing $QC1$.}

From the definition of $\otimes$ (\defref{def:composition}) it follows that transition labels in the $M_c[\var]$ used in $QC1$ have form $\dia{A} \land \gamma'[\var]$ for some $A \subseteq \AP$ and propositional query $\gamma'[\var]$.  This fact enables an immediate simplification of the computation of shattering intervals for edge labels as perfomed in \textit{ShatterFQA}.  Define two propositional queries $\gamma_1[\var]$ and $\gamma_2[\var]$ to be \emph{logically equivalent} if and only if for every propositional formula $\gamma$, $\gamma_1[\gamma] \equiv \gamma_2[\gamma]$.  We have the following.

\begin{lemma}
If propositional query $\gamma[\var] = \dia{A} \land \gamma'[\var]$, then $\gamma[\var]$ is logically equivalent to one of the following:
  \begin{itemize}
    \item $\dia{A}$
    \item $\logicfalse$
    \item $\dia{A} \land \var$
    \item $\dia{A} \land \lnot\var$
  \end{itemize}
\end{lemma}
\begin{proof}
$\dia{A}$ may be seen as a truth assignment to every $a \in \AP$ in the following sense:  $a$ is true if and only if $a \in A$.  We may then simplify $\gamma'[\var]$ as follows:  replace every atomic proposition $a$ in $\gamma'[\var]$ by the proposition $\logictrue$ if $a \in A$, and $\logicfalse$ otherwise, then simplify the resulting expression, which only contains propositional connectives, $\var$, $\logictrue$ and $\logicfalse$.   It is straightforward to prove that this results in one of the listed formulas.
\end{proof}

This fact simplifies testing edge-label shatterability and edge deadness, as required in \textit{ShatterFQA}; an edge is dead if and only if its label simplifies to $\logicfalse$, and is shatterable if and only if it simplifies to $\dia{A} \land \var$ or $\dia{A} \land \lnot\var$.  It also means that all shatterable edge labels in $M_{c}[\var]$ have shattering intervals of the form $[\lnot\dia{A}, \logicfalse]$ (if the label simplifies to $\dia{A} \land \var$) or $[\logictrue, \dia{A}]$ (if simplifies to $\dia{A} \land \lnot\var$).  Thus no conversions to shattering normal form (SNF) are needed on the composed edge labels.

We also remark on another optimization made possible by the restricted form of transition labels in $M_c$.  Recall that in \textit{ShatterFQA} the joint shattering condition $I_S$ for a set of transition labels is computed (line~\ref{alg-line:I_S} of Algorithm~\ref{alg:ShatterFQA}).  When the FQA being processed is $M_c$ we know that each of individual shattering conditions has form either $[\lnot\dia{A}, \logicfalse]$ or $[\logictrue, {\dia{B}}]$, where $A, B \subseteq \AP$.  Applying the conjunction operation as presented in Lemma~\ref{lem:interval-conjunction} to these types of shattering conditions yields an interval of the following form.
$$
I_S = [\lnot\dia{A_1} \land \cdots \land \lnot\dia{A_k}, \dia{B_1} \lor \cdots \lor \dia{B_\ell}]
$$
This composition interval may be represented computationally as two sets, $\mathcal{E} = \{A_1, \ldots, A_k\}$ and $\mathcal{F} = \{B_1, \ldots, B_\ell\}$, and that these sets have the property that for any $\gamma \in \Gamma$,
$
\gamma \in I_S
$
iff $\den{\gamma} \cap \mathcal{E} = \emptyset$ and $\mathcal{F} \subseteq \den{\gamma}$.  This implies a simple check for whether $\mathcal{I_S} = \emptyset$ holds:  $\mathcal{I_S} = \emptyset$ is true iff $\mathcal{E} \cap \mathcal{F} \neq \emptyset$.
Finally, Lemma~\ref{lem:epsilon} indicates that sets $[\logictrue]_{\sim_\varepsilon}$ and $[\logicfalse]_{\sim_\varepsilon}$ have representations matching the shattering conditions for transition labels in $M_c$.  Specifically, $[\logictrue]_{\sim_\varepsilon} = [\logictrue, \dia{\emptyset}]$, while $[\logicfalse]_{\sim_\varepsilon} = [\lnot\dia{\emptyset}, \logicfalse]$.  This fact can be used to implement efficiently the tests in lines~\ref{alg-line:if-true} and~\ref{alg-line:if-false} of Algorithm~\ref{alg:ShatterFQA}.  In particular, line~\ref{alg-line:if-true} requires checking if $I_S \cup [\logictrue]_{\sim_\varepsilon} = \emptyset$, where $I_S$ is the conjunction of shattering intervals for a set of transitions labels.  If the FQA in question is $M_c$ then $I_S$ can be represented as a pair $(\mathcal{E}, \mathcal{F})$ of subsets of $2^{\AP}$.  Testing non-emptiness of $I_S \cup [\logictrue]_{\sim_\varepsilon}$ when $I_S \neq \emptyset$ amounts to checking whether or not $\emptyset \in \mathcal{E}$; if it is, the non-emptiness check fails, while if it is not then it succeeds.  Likewise, $I_S \cap [\logictrue]_{\sim_\varepsilon}$ can be represented using the same pair-of-sets approach for $I_S$:  the representation becomes $(\mathcal{E}, \mathcal{F} \cup \{ \emptyset \})$.  Similar observations hold for the check in line~\ref{alg-line:if-false}.  These observations also enable empty intervals to be efficiently removed from $SC$.

\subsection{Query Checking Multiple Data Streams}

In this section we briefly discuss an iterative method for solving the query-checking problem in the case of multiple data streams.  Note that a propositional formula can be a solution to such a problem if and only if it is a solution for each of the single-stream problems formed by considering each data stream in isolation.  This suggests the following iterative strategy for the multi-stream case:  process each individual stream, one at a time, using the $QC1$ technique outlined above, then conjoin each interval in the result with a each interval in a running set of shattering intervals that are the solution for the previously considered streams.  Algorithm \ref{alg:multiqc} presents the approach.  Note that we have ``in-lined'' code for $QC1$ rather than calling it explicitly so that we can re-use the construction of $M_{\lnot\phivar}$ across the different data streams in $\Pi$.

\begin{algorithm}[t]
  \SetKwInOut{Input}{Input}
  \SetKwInOut{Output}{Output}
  \SetKwBlock{Preprocessing}{Preprocessing step:}{} 
  \SetKwBlock{Main}{Main:}{} 
  \DontPrintSemicolon 

  \underline{\textbf{Algorithm} \textit{QC} $(\Pi, \phivar)$}\;
  \Input{Set of $k$ data streams $\Pi = \{\pi_{i}\}$, Finite LTL Query $\phivar$}
  \Output{Set $\mathit{SC}$ of query solutions}
  
  \Preprocessing{
    $M_{\lnot\phivar}$ := FQA representation of $\lnot\phivar$ \tcp*{Tableau method}


  } 
  \Main{
    $SC := \{[\logictrue, \logicfalse]\}$

    \ForEach{$\pi \in \Pi$}{
      $M_{\pi} :=$ PNFA for $\pi$  \tcp*{Build automaton $M_\pi$}
      $M_{c}[\var] := M_{\lnot\phivar} \otimes M_{\pi}$ \tcp*{Compose FQA for $\lnot\phivar$ and $M_\pi$}
      $SC' := \mathit{ShatterFQA}(M_{c}[\var])$ \tcp*{Compute shattering intervals}
      
      $SC := \bigl\{ [\gamma_1, \gamma_1'] \land [\gamma_2, \gamma_2'] \mid [\gamma_1, \gamma_1'] \in SC,  [\gamma_2, \gamma_2'] \in SC'\bigr\}$\; \tcp*{Update aggregate solution}

      \uIf{$SC$ contains only empty intervals \label{qc-line:early-termination}}{
        \Return $\emptyset$
      }
    }

    \Return $SC$
  } 
  \caption{Algorithm for query checking multiple data streams}\label{alg:multiqc}

\end{algorithm}

The algorithm $QC$ can obviously benefit from the same optimizations that we discussed for $QC1$.  The algorithm can also perform better if an appropriate ordering is used for processing the data streams in $\Pi$; in particular, it may generally be advantageous to process shorter data streams first.  Finally, the test in line~\ref{qc-line:early-termination} can be simplified if the intersection operation used to update $SC$ in the previous line ensures that no empty intervals are present in $SC$, using the ideas discussed for optimizing $QC1$.  In this case, it suffices to check whether or not $SC = \emptyset$ in order to determine whether $QC$ can terminate early or not.

\section{Experimental Results}
We have implemented our query-checking algorithms in C/C++.
Our tool allows the user to specify a Finite LTL query $\phivar$ and finite set of data streams $\Pi$, and outputs the set of solutions for $\var$ in the form of a set of propositional intervals.
The Spot~\cite{duret.16.atva2} platform (v.2.8.1) is used to handle the parsing of Finite LTL formulas, with an extension to support the $\tlWeakX$ (Weak Next) operator.
Several Spot-provided automatic formula rewrites (such as $\tlX \logictrue \equiv \logictrue$) are based on standard LTL identities and do not hold for Finite LTL; these are disabled.
The SymPy Python symbolic-computing package~\cite{meurer:pcs2017} is used to manipulate propositional formulas.



We conducted a set of experiments with two goals in mind: (1) to establish that meaningful results are produced; and (2) to perform a ``stress test'' to study the system's performance limitations.
For these purposes we used a synthetic dataset from the Causality Workbench Benchmark~\cite{causality_dataset} representing time-series data for product sales and promotions.
In the dataset, daily sales volumes for 100 different (anonymized) products were reported over a three-year period (1065 days total).
Over the same time period, the daily status (active or inactive) was recorded for 1000 different promotional campaigns.
Each product can be influenced by up to 50 different promotions in the dataset, with each promotion having a varying degree of influence on each product is targeted at.
An underlying 1000x100 correlation (influence) matrix between promotions and products is provided as part of the dataset.  Additionally, each product has a defined seasonal baseline.

To adapt this dataset for query checking, we created the following 1100 atomic propositions.
\begin{align*}
  & prod_{i} && i \in \{1,\ldots, 100\}\\
  & promo_{i} && i \in \{1,\ldots, 1000\}
\end{align*}
Intuitively, at time point $t$, $prod_{i}$ is true if product $i$'s sale volume on day $t$ was reported to be greater than the volume on day $t-1$, and false otherwise.  Similarly, at time point $t$, $promo_{i}$ is true if and only if the promotion was active on day $t$.  This construction yields a single data stream $D$ of 1095 time points, with $|\AP| = 1100$.  We also subdivide $D$ temporally into quarters to create a set of data streams $\dquarter$ (12 streams total, one for each quarter of each of the three years).
Additionally, we perform this division at a monthly level to create $\dmonth$ (36 total streams, one per month over three years).
Notationally, for a subset $A \subseteq \AP$, we denote by $D[A]$ ($\dquarter[A]$, $\dmonth[A]$) the data stream(s) in the given sets restricted to the set of atomic propositions $A$.

Two distinct query classes were explored. The first class has the form
$$\Phi_{1}[\var] = \{\tlG(prod_{i} \to \tlF (\var \land \tlX \logictrue))\}$$

\noindent for $i \in \{1,\ldots, 100\}$.
A query in this class is intended to characterize the impact that an increase in product $i$'s sales will have on future sales of other products, as specified by solutions to $\var$.
These queries are insensitive to promotional activity.
The second class has the form
$$\Phi_{2}[\var] = \{\tlG((promo_{i} \land prod_{j}) \to \tlF (\var \land \tlX \logictrue))\}$$
\noindent for $i \in \{1, \ldots, 100\}$ and $j \in \{1, \ldots, 1000\}$.
An instance of such a query captures the joint effect that promotion $i$ and an increase in product $j$'s sales have on the future sales of other products (encoded as solutions to $\var$)
This query class could be useful when planning a promotional strategy; for example, it can be used in assessing what additional effects a promotion targeting a specific product might have on other products as well.

The remainder of this section gives more details of the experiments we conducted to address the two goals defined above.
All experimental trials were carried out on a machine with an Intel Core i5-6600K, 32 GB RAM and a 64-bit version of GNU/Linux. 

\subsection{Benchmark 1 -- Detecting Meaningful Query Solutions}
We first illustrate our methodology's ability to produce comprehensible (to humans) solutions over smalls sets of atomic propositions.
To this end, we selected five sets of three atomic propositions each, with each such set containing two product propositions and one promotion proposition.  The selection criterion was to use the the top five promotions, ranked by the combined influence of their two most positively influenced products, as determined by the influence matrix belonging to the dataset.  This set $\mathcal{T}$ is: 
\begin{itemize}
  \item $\{prod_{27}, prod_{95}, promo_{233}\}$,
  \item $\{prod_{19}, prod_{96}, promo_{748}\}$,
  \item $\{prod_{19}, prod_{99}, promo_{710}\}$,
  \item $\{prod_{33}, prod_{52}, promo_{348}\}$,
  \item $\{prod_{44}, prod_{85}, promo_{433}\}$.
\end{itemize}
Then, for each $T \in \mathcal{T}$, we computed $D[T]$, $\dquarter[T]$, and $\dmonth[T]$.
Four queries were constructed, using atomic propositions from $T$ to form two instances of query classes $\Phi_{1}[\var]$ and two instances of $\Phi_{2}[\var]$, or four in total.
For example, with $T = \{prod_{27}, prod_{95}, promo_{233}\}$ we had the following queries:
\begin{itemize}
  \item $\tlG(prod_{27} \to \tlF(\var \land \tlX \logictrue))$,
  \item $\tlG(prod_{95} \to \tlF(\var \land \tlX \logictrue))$,
  \item $\tlG((promo_{233} \land prod_{27}) \to \tlF(\var \land \tlX \logictrue))$,
  \item $\tlG((promo_{233} \land prod_{95}) \to \tlF(\var \land \tlX \logictrue))$.
\end{itemize}
$\qc1(D[T], \phi[\var])$, $\qc(\dquarter[T], \phi[\var])$, and $\qc(\dmonth[T], \phi[\var])$ were computed for each of these four queries.
Solving the individual queries for these configurations took under a second of wall time each.
Shattering solutions were found for each of the queries in the single-stream $\qc 1(D[T], \phi[\var])$ scenario.
With $T = \{prod_{19}, prod_{96}, promo_{748}\}$ and $\phi[\var] = \tlG(prod_{19} \to \tlF(\var \land \tlX \logictrue))$, for example, the solution of $\qc1(D[T], \phi[\var])$ was found to be
$$\{ [\logictrue, prod_{19} \land \lnot prod_{96}], [\logictrue, \lnot prod_{19} \land \lnot prod_{96}]\},$$
suggesting that the rise of sales for $prod_{19}$ leads to an eventual decline of sales of $prod_{96}$.  When considering a query of the second class instead, i.e. $\phi[\var] = \tlG((promo_{748} \land prod_{19}) \to \tlF(\var \land \tlX \logictrue))$, the solution was found to be
$$\{ [\logictrue, prod_{19} \land prod_{96} \land promo_{748}], [\logictrue, \lnot prod_{19} \land \lnot prod_{96} \land promo_{748}]\}.$$
This indicates that the promotion either has an eventual positive impact, or negative impact, on the eventual sales of both products.

Over-all, solutions to queries of class $\Phi_{2}$ were observed to be related to solutions of the corresponding query of class $\Phi_{1}$.
Similar shattering intervals were computed for queries $\phi_{1}[\var]$ and $\phi_{2}[\var]$ for a given $T$.

To account for the seasonal variation known to exist in the dataset, we also performed query checking for data streams grouped temporally: $\dquarter$ was subdivided into four subsets, one each for the winter, spring, summer and fall quarters of the three years in the data set.
We similarly divided $\dmonth$ is subdivided into 12 subsets as well, one for each month.
The solutions to the relevant query-checking problems were then computed using these new data-stream groupings.  
For sets
  $\{prod_{27}, prod_{95}, promo_{233}\}$,
  $\{prod_{19}, prod_{99}, promo_{710}\}$, and
  $\{prod_{33}, prod_{52}, promo_{348}\}$
it was observed that certain solutions for $\phi_{2}[\var]$ characterized only a portion of the corresponding solution for $\phi_{1}[\var]$: when $T = \{prod_{19}, prod_{99}, promo_{710}\}$, the solution to the query $\phi_{1}[\var] = \tlG(prod_{19} \to \tlF(\var \land \tlX \logictrue))$ when only considering streams from the 11th month of each year was found to be
$$\{[\logictrue, \lnot prod_{19} \land prod_{99}]\}$$
while the solution for $\phi_{2}[\var] = \tlG((promo_{710} \land prod_{19}) \to \tlF(\var \land \tlX \logictrue))$ was
$$\{[\logictrue, \lnot prod_{19} \land prod_{99} \land \lnot promo_{710}], 
    [\logictrue, prod_{19} \land prod_{99} \land promo_{710}]\}.$$
For the other listed sets in $\mathcal{T}$ this type of variation was noted for both the full data stream, quarterly, and monthly breakdowns.

\subsection{Benchmark 2 -- Performance Evaluation}
We then assessed the overall performance of our implementation using the different sets of data streams and also differently-sized sets of atomic propositions.  In general, as the set of atomic propositions grows so does the cost of query checking; our goal was to determine how long it would take for our implementations to reach a ``time out'' of one hour of computing time.

More specifically,
a fixed number of atomic propositions corresponding to products and promotions were sampled as set $A \subset \AP$.
The restricted stream $D[A]$, as well as the restricted stream sets $\dquarter[A]$ and $\dmonth[A]$ were then checked against instances of the query classes defined above:  an instance $\phi_{1}[\var] \in \Phi_{1}[\var]$ is chosen by selecting $prod_{i} \in A$, similarly an instance $\phi_{2}[\var] \in \Phi_{2}[\var]$ wass chosen by selecting $prod_{i}, promo_{j} \in A$.

Performance results are presented for computing the single-stream query checking problems $\qc 1(D, \phi)$, $\qc(\dquarter[A], \phi)$, and $\qc(\dmonth[A], \phi[\var])$ for $\phi \in \{\phi_{1}[\var], \phi_{2}[\var]\}$.  Experiments are repeated by re-sampling $A$, with and statistical averages reported over all experiments in the same category. 
 
\figref{fig:ltlqc-average-samples} shows the performance results for $\qc 1(D[A], \phi[\var])$.
Average computation time is reported, along with the number of shatterable edges and distinct shatterable labels in the composed FQA produced during query checking.
\figref{fig:ltlqc-average-samples-quarterly} reports the same statistics for $\qc (\dquarter[A], \phi[\var])$, as does \figref{fig:ltlqc-average-samples-monthly} for $\qc (\dmonth[A], \phi[\var])$.  In these cases, the average running time of the individual calls to $\qc1$ made by each of the calls to $\qc$ is reported.
Individual invocations of $\qc1$ that run for over 1 hour are terminated and aggregates affected are reported as timeout (``t/o'').
For $\dquarter[A]$ (as well as $\dmonth[A]$), we note that it could be interesting to compute a new $A$ for each new $\qc1$ computation (as $A$ is the source of experimental randomness) rather than computing single-stream query checking for \emph{all} streams in $\dquarter[A]$.


\begin{figure}[!htb]
  \centering
  \begin{tabular}{|c|c|c|c|c|c|}
    \hline
    \# products & \# promotions & $\phi_{1}[\var]$ avg time & avg shatterable edges & avg shatterable labels \\
    \hline
    2 & 0 &  $2.387 \pm 0.18$ & $1642.4 \pm 9.7$ & $4.0\pm 0.0$ \\
    3 & 0 &  $3.460 \pm 0.42$ & $1640.5 \pm 8.3$ & $8.0 \pm 0.0$ \\
    4 & 0 &  $17.35 \pm 6.01$ & $1640.4 \pm 8.3$ & $16.0 \pm 0.0$ \\
    5 & 0 &  t/o & t/o & t/o \\
    2 & 1 &  $3.311 \pm 0.19$ & $1642.2 \pm 9.6$ & $7.7 \pm 0.94$ \\
    3 & 1 &  $14.58 \pm 6.70$ & $1641.9 \pm 8.9$ & $15.2 \pm 2.1$ \\
    4 & 1 &  t/o & t/o & t/o \\
    2 & 2 &  $13.01 \pm 6.27$ & $1640.9 \pm 9.7$ & $14.5\pm 2.4$ \\
    3 & 2 &  t/o & t/o & t/o \\
    2 & 3 &  t/o & t/o & t/o \\
    \hline
  \end{tabular}

  \begin{tabular}{|c|c|c|c|c|c|}
    \hline
    \# products & \# promotions & $\phi_{2}[\var]$ avg time & avg shatterable edges & avg shatterable labels \\
    \hline
    2 & 1 & $2.741 \pm 0.44$ & $1326.2 \pm 253.7 $ & $7.8\pm 0.87$ \\
    3 & 1 & $8.511 \pm 5.73$ & $1379.4 \pm 205.9$ & $15.3\pm 2.1$ \\
    4 & 1 & t/o & t/o & t/o \\
    2 & 2 & $7.627 \pm 5.34$ & $1320.7 \pm 266.2$ & $14.6\pm 2.4$ \\
    3 & 2 & t/o & t/o & t/o \\
    2 & 3 & t/o & t/o & t/o \\
    \hline
  \end{tabular}

  \caption[Runtime Performance of Finite LTL Query Checking]{Average performance of $\qc 1(D[A], \phi[\var])$.  $\phi_{2}[\var]$ is not applicable when no promotions are present in $A$.  Data stream is of length 1095.  Averages taken over 100 samples.  All times in seconds.}\label{fig:ltlqc-average-samples}

\end{figure}

\begin{figure}[!htb]
  \centering
  \begin{tabular}{|c|c|c|c|c|c|}
    \hline
    \# products & \# promotions & $\phi_{1}[\var]$ avg time & avg shatterable edges & avg shatterable labels \\
    \hline
    2 & 0 &  $0.527 \pm 0.03$ & $134.7 \pm 4.3$ & $4.0 \pm 0.0$ \\
    3 & 0 &  $0.632 \pm 0.03$ & $135.1 \pm 3.9$ & $8.0\pm 0.0$ \\
    4 & 0 &  $1.537 \pm 0.39$ & $134.8 \pm 3.8$ & $16.0\pm 0.13$ \\
    5 & 0 &  t/o & t/o & t/o \\
    2 & 1 &  $0.609 \pm 0.04$ & $135.4 \pm 3.5$ & $6.6 \pm 1.8$ \\
    3 & 1 &  $1.115 \pm 0.41$ & $135.4 \pm 3.8$ & $12.9 \pm 3.5$ \\
    4 & 1 &  t/o & t/o & t/o \\
    2 & 2 &  $0.760 \pm 0.23$ & $135.2 \pm 3.6$ & $7.7 \pm 3.2$ \\
    3 & 2 &  t/o & t/o & t/o \\
    2 & 3 &  $35.00 \pm 202.8$ & $134.8 \pm 3.7$ & $13.0\pm 5.2$ \\
    \hline
  \end{tabular}

  \begin{tabular}{|c|c|c|c|c|c|}
    \hline
    \# products & \# promotions & $\phi_{2}[\var]$ avg time & avg shatterable edges & avg shatterable labels \\
    \hline
    2 & 1 & $0.541 \pm 0.04$ & $113.2 \pm 24.4$ & $6.3 \pm 1.8$ \\
    3 & 1 & $0.765 \pm 0.23$ & $109.2 \pm 35.8$ & $12.2\pm 3.1$ \\
    4 & 1 & $464.60 \pm 1482.4$ & $115.4 \pm 33.7$ & $20.6\pm 5.4$ \\
    2 & 2 & $0.650 \pm 0.12$ & $111.3 \pm 29.4$ & $8.8 \pm 3.3$ \\
    3 & 2 & $112.86 \pm 568.2$ & $113.86\pm 30.2$ & $17.71 \pm 5.71$ \\
    2 & 3 & $78.05 \pm 636.21$ & $107.205 \pm 35.8$ & $12.4 \pm 6.02$ \\
    \hline
  \end{tabular}

  \caption[Runtime Performance of Finite LTL Query Checking]{Performance of $\qc 1(\pi \in \dquarter[A], \phi[\var])$.  Data streams are of length 91 or 92.  Averages taken over 60 trials.  All times in seconds.}\label{fig:ltlqc-average-samples-quarterly}

\end{figure}

\begin{figure}[!htb]
  \centering
  \begin{tabular}{|c|c|c|c|c|c|}
    \hline
    \# products & \# promotions & $\phi_{1}[\var]$ avg time & avg shatterable edges & avg shatterable labels \\
    \hline
    2 & 0 &  $0.399 \pm 0.02$ & $43.9 \pm 2.4$ & $4.0 \pm 0.0$ \\
    3 & 0 &  $0.442 \pm 0.02$ & $43.9 \pm 2.4$ & $7.9 \pm 0.3$ \\
    4 & 0 &  $0.546 \pm 0.06$ & $43.8 \pm 2.4$ & $13.9 \pm 1.1$ \\
    5 & 0 &  $5.248 \pm 7.27$ & $43.7 \pm 2.4$ & $19.8 \pm 1.7$ \\
    6 & 0 &  $140.6 \pm 286.0$ & $43.8 \pm 2.7$ & $24.4 \pm 2.1$ \\
    7 & 0 &  $619.1 \pm 974.9$ & $44.0 \pm 2.3$ & $26.7 \pm 1.8$ \\
    2 & 1 &  $0.412 \pm 0.02$ & $43.8 \pm 2.4$ & $5.4 \pm 1.7$ \\
    3 & 1 &  $0.487 \pm 0.30$ & $43.6 \pm 2.4$ & $9.2 \pm 2.2$ \\
    4 & 1 &  $1.445 \pm 2.27$ & $43.7 \pm 2.6$ & $15.5 \pm 2.6$ \\
    2 & 2 &  $0.444 \pm 0.40$ & $43.5 \pm 2.5$ & $6.3 \pm 2.4$ \\
    3 & 2 &  $0.668 \pm 0.60$ & $43.8 \pm 2.6$ & $11.6 \pm 3.3$ \\
    2 & 3 &  $0.504 \pm 0.12$ & $43.7 \pm 2.5$ & $8.4 \pm 3.1$ \\
    \hline
  \end{tabular}

  \begin{tabular}{|c|c|c|c|c|c|}
    \hline
    \# products & \# promotions & $\phi_{2}[\var]$ avg time & avg shatterable edges & avg shatterable labels \\
    \hline
    2 & 1 & $0.372 \pm 0.02$ & $37.1 \pm 10.6$ & $5.5 \pm 1.6$ \\
    3 & 1 & $0.455 \pm 0.04$ & $38.6 \pm 9.8$ & $10.6 \pm 2.8$ \\
    4 & 1 & $3.547 \pm 9.24$ & $40.5 \pm 8.2$ & $16.0 \pm 3.8$ \\
    2 & 2 & $0.392 \pm 0.02$ & $39.5 \pm 7.9$ & $5.9 \pm 1.4$ \\
    3 & 2 & $1.571 \pm 5.11$ & $38.4 \pm 8.8$ & $13.1\pm 4.4$ \\
    2 & 3 & $0.493 \pm 0.14$ & $41.9 \pm 7.0$ & $7.6 \pm 3.1$ \\
    \hline
  \end{tabular}

  \caption[Runtime Performance of Finite LTL Query Checking]{Performance of $\qc 1(\pi \in \dmonth[A], \phi[\var])$.  Data streams are of length 30 or 31.  Averages taken over 360 trials.  All times in seconds.}\label{fig:ltlqc-average-samples-monthly}

\end{figure}

It can be observed that the time cost of query checking increases with both the number of atomic propositions considered ($|A|$), as well as with the length of the data stream.  The number of shatterable edges present in the composed FQA also is (perhaps not surprisingly) correlated with data-stream length.  The number of shatterable labels (i.e. $|SL|$ from Algorithm \ref{alg:ShatterFQA}) has in the worse case an exponential relationship with runtime and so efforts to lower it (primarily by restricting $|A|$) will afford runtime savings.
We also note that there is a potential for the tradeoff of supporting a larger set of atomic propositions if the data streams can be shortened in length.  This is especially noticeable for breakdowns of $\dmonth[A]$, where sets $A \subset \AP$ with $|A| = 6$ and $|A| = 7$ are processable.

\section{Conclusions and Future Work}
We have presented our work on performing LTL query checking over normalized finite data
streams using an automaton-theoretic approach.  Solving such queries can aid in
the comprehension of system behavior underlying observed execution traces,
which might allow for one to better understand or diagnose a system when it
cannot be accessed directly.

In terms of future work, more experimentation with our techniques would yield more insight into their applicability and use cases.  In particular, using datasets consisting of system logs or financial data would help gauge our methods' utility in these data-rich areas.  We would also like to investigate linear temporal logics that incorporate reasoning about discrete time, such as a discrete-time version of MITL~\cite{alur1996benefits}, as a basis for query checking over non-normalized finite data streams, as well as other, more expressive linear-time logics such as Linear Dynamic Logic~\cite{degiacomo:ijcai2013} and Allen Linear Temporal Logic~\cite{rocsu2006allen}.

\bibliography{local}
\end{document}